\documentclass[11pt]{article}
\usepackage{graphicx,color}
\usepackage{url}
\usepackage{here}
\usepackage{bm}
\usepackage{amssymb}
\usepackage{amsmath}
\usepackage{amsthm}
\usepackage{amsfonts}
\usepackage{mathtools,stmaryrd}
\usepackage{mathrsfs}  
\usepackage{float} 

\usepackage{todonotes}

\usepackage{hyperref,nameref}
\usepackage[english]{babel}

\usepackage{tikz,tikz-3dplot}
\usepackage{geometry}
\usepackage{combelow}

\usepackage{tikz}
\usetikzlibrary{shapes.geometric, arrows}

\tikzstyle{block} = [rectangle, draw, text width=8em, text centered, rounded corners, minimum height=3em]
\tikzstyle{line} = [draw, -latex']

\theoremstyle{definition}

\newtheorem{thm}{Theorem}

\newtheorem{lemma}[thm]{Lemma}

\newtheorem{definition}[thm]{Definition}

\newtheorem{remark}[thm]{Remark}

\graphicspath{{./figure/}} 

\begin{document}\sloppy

\newcommand{\SO}{\mathrm{SO}}
\newcommand{\GL}{\mathrm{GL}}
\newcommand{\so}{\mathfrak{so}}
\newcommand{\Liederi}{\mathcal{L}} 
\newcommand{\Liexi}{\mathcal{L}_{\xi}} 
\newcommand{\Liedualxi}{\mathcal{L}^{\ast}_{\xi}} 
\newcommand{\bbbr}{\mathbb{R}}
\newcommand{\bbbrn}{\mathbb{R}^n}
\newcommand{\bbbrrr}{\mathbb{R}^3}
\newcommand{\bbbrrrmat}{\mathbb{R}^{3\times 3}}
\newcommand{\tr}{\mathrm{tr}}
\renewcommand{\d}{\mathrm{d}}

\newcommand{\cay}{\mathrm{Cay}}
\newcommand{\const}{\mathrm{const}}
\newcommand{\Fe}{\mathcal{F}_\mathrm{E}}
\newcommand{\Fb}{\mathcal{F}_\mathrm{V}}
\newcommand{\mdef}{:=}
\newcommand{\ex}{\mathrm{e}_x}
\newcommand{\ez}{\mathrm{e}_z}
\newcommand{\ad}{\mathrm{ad}}
\newcommand{\Ad}{\mathrm{Ad}}
\newcommand{\diag}{\mathrm{diag}}
\newcommand{\andrm}{~ \mathrm{and} ~}
\newcommand{\grad}{\mathrm{grad}}
\newcommand{\Exp}{\mathrm{Exp}}
\newcommand{\lpartial}{\bar\partial}
\newcommand{\En}{\mathcal{E}}
\newcommand{\g}{\mathfrak{g}}
\newcommand{\gad}{\mathfrak{g}^*}

\newcommand{\inverse}{^{-1}}
\newcommand{\trans}{^{\top}}
\newcommand{\skewmap}{^{\times}}
\newcommand{\unskewmap}{^{\vee}}
\newcommand{\abs}[1]{\left|#1\right|}
\newcommand{\norm}[1]{\left\lVert#1\right\rVert}
\newcommand{\normtwo}[1]{\left\lVert#1\right\rVert_2}
\newcommand{\normf}[1]{\left\lVert#1\right\rVert_F}
\newcommand{\normx}[1]{\left\lVert#1\right\rVert_x}
\newcommand{\normxk}[1]{\left\lVert#1\right\rVert_{x_k}}
\newcommand{\inpro}[1]{\left\langle#1\right\rangle}

\newcommand{\mathcolorbox}[2]{\colorbox{#1}{$\displaystyle #2$}}

\title{Euler--Poincar\'e reduction and the Kelvin--Noether theorem for discrete  mechanical systems with advected parameters and additional dynamics}

\author{Yusuke Ono$^{1}$\footnote{Email: yuu555yuu@keio.jp}, ~ Simone Fiori$^{2}$\footnote{Email: s.fiori@staff.univpm.it} ~ and Linyu Peng{$^3$}\footnote{Corresponding author. Email: l.peng@mech.keio.ac.jp}\vspace{0.4cm}\\
{\it 1. Graduate School of Science and Technology, Keio University,} \\
{\it Yokohama 223-8522, Japan}\\
{\it 2. Department of Information Engineering, Marche Polytechnic University,} \\
{\it Ancona I-60131, Italy}\\
{\it 3. Department of Mechanical Engineering, Keio University,} \\
{\it Yokohama 223-8522, Japan}\\
 }

\maketitle
\begin{abstract}
    The Euler--Poincar\'e equations, firstly introduced by Henri Poincar\'e in 1901, arise from the application of Lagrangian mechanics to systems on Lie groups that exhibit symmetries, particularly in the contexts of classical mechanics and fluid dynamics. 
    These equations have been extended to various settings, such as semidirect products, advected parameters, and field theory, and have been widely applied to mechanics and physics.
    In this paper, we introduce the discrete Euler--Poincar\'e reduction for discrete Lagrangian systems on Lie groups with advected parameters and additional dynamics, utilizing  group difference maps. Specifically, the group difference map is defined using either the Cayley transform or the matrix exponential.
    The continuous and discrete Kelvin--Noether theorems are extended accordingly, that account for Kelvin--Noether quantities of the corresponding continuous and discrete Euler--Poincar\'e equations. 
    As an application, we show both continuous and discrete Euler--Poincar\'e formulations about the dynamics of underwater vehicles, followed by numerical simulations.
   Numerical results illustrate the scheme's ability to preserve geometric properties over extended time intervals, highlighting its potential for practical applications in the control and navigation of underwater vehicles, as well as in other domains.
\end{abstract}

\section{Introduction}
Many physical systems exhibit symmetries, such as invariance under translation, rotation, or particle relabeling.
Symmetries can be used to simplify the dynamics by reducing the degrees of freedom or the order of an invariant system (e.g., \cite{marsden2013introduction,olver1993}).
In Lagrangian mechanics,  the symmetries of a Lagrangian, specifically its invariance under Lie group actions, enable the reduction of the corresponding Euler–Lagrange equations (i.e., the equations of motion), for example, through the use of invariantized coordinates. More generally, this can be extended to symmetries of functionals, also known as least actions.
One of the reduction methods, called Euler--Poincar\'e reduction, leads to the Euler--Poincar\'e equations, is a particular case of Lagrangian reduction when the configuration space is a Lie group $G$ and the Lagrangian is $G$-invariant \cite{marsden1993lagrangian, marsden1993reduced}.
It has been widely applied in rigid body dynamics \cite{leok2007overview}, fluid mechanics \cite{Arnold1966a}, 
magnetohydrodynamics \cite{Holm129518}, 
plasma physics \cite{Holm174831, squire2013hamiltonian}, 
thermodynamical simple systems \cite{coueraud2018variational}, etc. 
The Euler--Poincar\'e equations can also be expressed in terms of  the Lie--Poisson equations via the reduced Legendre transformation. 

Reduction theory in discrete mechanics is also often approached from a variational perspective, as demonstrated in \cite{veselov1991discrete, veselov1988integrable}.
These variational integrators are known to play a crucial role in the numerical integration of mechanical systems and geometric mechanics \cite{kane2000variational, Wendlandt1997223}.
Some pioneering work on discretizations of reduction theory includes \cite{bobenko1999discrete,GeMarsden1988134}. 
In particular, the discrete analogue of the Euler--Poincar\'e reduction has been studied in \cite{gawlik2011geometric,marsden1999discrete, pavlov2011structure}. Implementing numerical iterations of the discrete Euler--Poincaré equations in the (linear) Lie algebra rather than directly using the Lie group-valued Euler--Lagrange equations involves a discretization strategy that focuses on the algebraic structure, simplifying the computations by working with the Lie algebra rather than the more complex Lie group.

The Euler--Poincar\'e reduction has been investigated in various contexts (e.g., \cite{cendra_holm_marsden_ratiu_1998, marsden2000reduction}). 
Recently, in \cite{HOLM2023133847}, the authors proposed (continuous) Euler--Poincar\'e reduction with advected parameters and additional dynamics, which was applied to study the dynamics of the wave mean flow interaction.
In this paper, we develop the discrete analogue of the Euler--Poincar\'e reduction incorporated advected parameters and additional dynamics and then propose the corresponding discrete Euler--Poincar\'e equations with additional dynamics.

The contributions of this study are summarized below.
\begin{itemize}
    \item[(1)] The discrete Euler--Poincar\'e reduction for Lagrangian mechanics with advected parameters and additional dynamics is formulated.
    \item[(2)] The Kelvin--Noether theorem and its discrete analogue are generalized to systems with additional dynamics and advected parameters.
    \item[(3)] The discrete Euler--Poincar\'e reduction is applied to the dynamics of underwater vehicles by deriving both the continuous and discrete Euler--Poincar\'e equations.
    The behavior of the total energy and Kelvin--Noether quantity is also presented to demonstrate the efficiency of the scheme.
\end{itemize}

This paper is structured as follows. 
In Section \ref{sec: EPreduction}, we recall the Euler--Poincar\'e reduction for systems with and without advected parameters and additional dynamics and extend the Kelvin--Noether theorem to incorporate advected parameters and additional dynamics simultaneously.
In Section \ref{sec: DisEPreduction}, the discrete counterpart of the theories introduced in Section \ref{sec: EPreduction} is developed by using the discrete variational calculus. 
In Section \ref{sec: Application2UV}, the discrete Euler--Poincar\'e reduction with advected parameters and additional dynamics is applied to the dynamics of underwater vehicles. 
We derive the numerical schemes and the discrete Kelvin--Noether quantities in cases that the group difference map is chosen as either the Cayley transform or the matrix exponential. 
In Section \ref{sec: NumericalSimulation}, we present  numerical results for the dynamics of underwater vehicles, focusing on the behavior of the total energy and the Kelvin--Noether quantity.
Finally, we conclude and discuss directions of future research in Section \ref{sec: Conclusion}.

\section{The continuous Euler--Poincar\'e reduction} \label{sec: EPreduction}
Consider a mechanical system with a configuration space $Q$, an $n$-dimensional smooth manifold, 
and let $G$ be a Lie group acting smoothly on $Q$.
In the current study, left group actions will be considered unless otherwise specified. 
The action of a group element $g \in G$ on a point $q \in Q$ is denoted as $g q$, while its tangent lift to $v_q \in T_q Q$ is denoted by $g v_q$. 
Consider a smooth curve $q(t)\in Q $ for $t \in \left[0,T\right]$ and its tangent vector $\dot{q}(t)$.
In the following, the time variable $t$ is omitted for brevity.
A Lagrangian defined in the tangent bundle $TQ$, $L: TQ \to \bbbr$, is  $G$-invariant if
\begin{equation}
    L \left( g q, g \dot q\right) = L \left( q, \dot q\right), \ \forall g \in G, q \in Q, \andrm \ \dot q \in T_{q} Q.
\end{equation}
If $Q = G$ and choosing left action of $g\inverse$, the (left) reduced Lagrangian $\ell: \g \to \bbbr$ defined on the Lie algebra  can be derived as
\begin{equation}
    \ell(\xi) := L \left( g^{-1} g, g^{-1} \dot g\right) = L \left( e, \xi\right),
\end{equation}
where $\xi=g^{-1}\dot{g}\in \g$ and $e\in G$ is the identity element.
Associated with it, the `basic' Euler--Poincar\'e equations hold (e.g., \cite{bloch1996euler, marsden1993lagrangian, marsden1993reduced,Poincare1901})
\begin{equation} \label{eq: basicEP}
    \frac{d}{dt} \frac{\partial \ell}{\partial \xi} = \ad^*_{\xi} \frac{\partial \ell}{\partial \xi},
\end{equation}
where the adjoint operator is defined by using the Lie bracket as 
\begin{equation}
\begin{aligned}
     \ad:\g\times \g &\to \g\\
     (\xi,\eta)&\mapsto \ad_{\xi}\eta:=[\xi,\eta],
\end{aligned}
\end{equation} and the coadjoint operator $\ad^*: \g \times \gad \to \gad$ for left action is defined as the pairing between $\g$ and its dual $\gad$,
\begin{equation}
    \inpro{\ad^*_{\xi} \mu,\eta} := \inpro{\mu, \ad_{\xi} \eta},\ \mathrm{for} \ \mu \in \gad, \ \xi, \eta \in \g.
\end{equation}
Importance of the Euler--Poincar\'e equations \eqref{eq: basicEP}
(e.g., its relation to nonholonomic constraints and its Hamiltonian character)
was first recognized by Hamel \cite{hamel1904, hamel1949} and Chetayev \cite{chetayev1941}.

Analogous to the equivalence of the Euler--Lagrange equations for a regular Lagrangian defined om the tangent bundle $TQ$  and Hamilton's equations on the cotangent bundle $T^*Q$ through the Legendre transformation, 
the Euler--Poincar\'e equations \eqref{eq: basicEP} are similarly equivalent to the Lie--Poisson equations
\begin{equation} \label{eq: LPeqplus}
    \dot{\mu} = \ad^*_{\frac{\partial h}{\partial \mu}} \mu,
\end{equation}
where the Legendre transformation between $\g$ and $\g^*$ defines
\begin{equation}
    \mu := \frac{\partial \ell}{\partial \xi},
\end{equation}
and $h$ is the reduced Hamiltonian defined on $\g^*$, 
\begin{equation}
h(\mu) := \inpro{\mu, \xi} - \ell(\xi).
\end{equation}
It follows immediately that 
\begin{equation}
\frac{\partial h}{\partial \mu} = \xi + \inpro{\mu, \frac{\partial \xi}{\partial \mu}} - \inpro{\frac{\partial \ell}{\partial \xi}, \frac{\partial \xi}{\partial \mu}} = \xi.
\end{equation}


\subsection{The Euler--Poincar\'e equations with advected parameters}
Holm et al. studied mechanical systems with advected parameters in a dual vector space $V^*$ within the Lagrangian semidirect product theory \cite{HOLM19981}.
However, it was also pointed out therein that the Euler--Poincar\'e equations with advected parameters are not the Euler--Poincar\'e equations derived from the semidirect product Lie algebra. 
 Advected parameters often appear, for instance, when a potential function in a Lie group $G$ is involved, 
and they bear dynamical significance in the reduction context.
For instance, in the case of the heavy top, the parameter corresponds to the unit vector aligned with gravity. 
In the context of compressible flows, it represents the fluid density in the reference configuration.
The following theorem for the Euler--Poincar\'e equations with advected parameters was shown in \cite{HOLM19981}.
\begin{thm}\label{thm: basic_EP}
    Let $G$ be a Lie group that acts on a dual vector space $V^*$ from the left. 
    Assume that the Lagrangian $L:TG \times V^* \to \mathbb{R}$ is left $G$-invariant. 
    When $a_0 \in V^*$, the corresponding function $L_{a_0}: TG \to \mathbb{R}$, 
    defined by $L_{a_0}(v_g) := L(v_g, a_0)$, is left $G_{a_0}$-invariant, 
    where $G_{a_0} \subseteq G$ denotes the isotropy group of $a_0$. 
    The $G$-invariance of $L$ allows one to define $\ell: \mathfrak{g} \times V^* \to \mathbb{R}$ by
    \begin{equation}
        \ell\left(g^{-1} v_g, g^{-1} a_0\right) = L(v_g, a_0).
    \end{equation}
    For a curve $g(t)$ on $G$, let $\xi(t) = g(t)^{-1} \dot{g}(t) \in \g$ and define a curve for the advected parameters $a(t) = g(t)^{-1} a_0\in V^*$. 
  Their variations satisfy
  \begin{equation}
      \delta \xi = \dot{\eta} +\ad_{\xi}\eta,\quad \delta a = - \eta a = - \mathcal{L}_{\eta} a,
  \end{equation}    
  where $\eta =g^{-1}\delta {g}$ and $\Liexi$ is the Lie derivative with respect to $\xi$. 
  Therefore, Hamilton's principle for the action $\int \ell dt$ leads to the (left-left\footnote{In the current study, left group actions of $G$ on $V^*$ are considered and the Lagrangians are assumed to be left invariant too.}) 
    Euler--Poincar\'e equations  on $\mathfrak{g} \times V^*$: 
    \begin{equation} \label{eq: conEPw/ad}
        \begin{aligned}
            \frac{d}{dt} \frac{\partial \ell}{\partial \xi}  = \ad^*_{\xi} \frac{\partial \ell}{\partial \xi} + \frac{\partial \ell}{\partial a} \diamond a,
        \end{aligned}
    \end{equation}
    and the curve $a(t)=g(t)^{-1}a_0$ is the unique solution of the linear differential equation 
       \begin{equation} 
        \begin{aligned}
            \frac{d}{dt} a = - \mathcal{L}_{\xi} a
        \end{aligned}
    \end{equation}
        with initial condition $a(0)=a_0$, 
     and the diamond operator $\diamond: V \times V^* \to \gad$ is the bilinear operator defined by
    \begin{equation}
        \left\langle b, \Liederi_{\xi} a \right\rangle _{V\times V^*} 
        := - \left\langle b \diamond a, \xi \right\rangle _{\gad \times \g},
    \end{equation}
    for all $b \in V$, $a \in V^*$, and $\xi \in \g$.
\end{thm}


\subsection{The Euler--Poincar\'e equations with advected parameters and additional dynamics}\label{sec: EP_w_additional}
In this section, we recall the Euler--Poincar\'e equations with advected parameters and additional dynamics introduced in \cite{HOLM2023133847}.
Consider $G \times Q$ as the configuration space and the curves $(g(t),q(t))$ in $ G\times  Q$ for $t \in \left[0,T\right]$ with tangent vectors $\dot g(t) \in T_{g}G$ and $\dot q(t) \in T_{q}Q$.
A Lagrangian defined in $T(G\times Q)$ with advected parameters $a_0\in V^*$ can be defined as
\begin{equation}
\begin{aligned}
    L: T\left(G\times Q\right) \times V^* &\to \mathbb{R}\\
      \left(g, \dot{g}, q, \dot{q}, a_0 \right) &\mapsto L\left(g, \dot{g}, q, \dot{q}, a_0 \right). 
\end{aligned}
\end{equation}
Invariance of the Lagrangian under the left action of $G$ reads
\begin{equation}
    L \left(g, \dot{g}, q, \dot{q}, a_0 \right) = L \left(\tilde{g} g, \tilde{g} \dot{g}, \tilde{g} q, \tilde{g} \dot{q}, \tilde{g} a_0 \right), \ \forall \tilde{g} \in G.
\end{equation} 
Assuming that $\tilde{g} = g^{-1}$, the reduced Lagrangian can be defined as the function $\ell: \mathfrak{g} \times TQ \times V^* \to \mathbb{R}$,
\begin{equation}
    \ell \left( \xi,n, \nu,  a \right) := L \left(g^{-1} g, g^{-1} \dot{g}, g^{-1} q, g^{-1} \dot{q}, g^{-1} a_0 \right),
\end{equation}
where $\xi := g^{-1} \dot{g} \in \g$, $n := g^{-1} q\in Q$,  $\nu := g^{-1} \dot{q} \in T_nQ$, and $a: =g^{-1}a_0\in V^*$. 
%
Defining $\eta := g^{-1} \delta g\in \g$ and $w := g^{-1} \delta q \in T_nQ$, 
the variations of $\xi,n, \nu$, and $a$ can be calculated as follows:
\begin{equation}\label{eq:variadd}
    \begin{aligned}
            \delta \xi   = \dot{\eta} + \ad_{\xi} \eta, \quad \delta n   =  w - \eta n, \quad 
        \delta \nu   = \dot{w} + \xi w - \eta \nu = \dot{w} + \mathcal{L}_{\xi} w - \mathcal{L}_{\eta} \nu, \quad
        \delta a    = - \mathcal{L}_{\eta} a ,
    \end{aligned} 
\end{equation}
where we have used the definitions $\mathcal{L}_\xi w :=\xi w$ and $\mathcal{L}_{\eta} a:= \eta a$.
Note that $\eta n\in T_n Q$  is defined by the infinitesimal generator at $n$, $\eta n := \eta_Q(n)$, or the Lie algebra action of  the manifold $Q$, 
as \cite{cendra2001lagrangian}
\begin{equation}
    \left(\eta , n\right) \in \g \times Q \mapsto \eta_Q(n) := \left.\frac{d}{ds}\right|_{s=0} \exp\left(s \eta \right) n.
\end{equation}
The variations $\eta$ and $w$ are arbitrary and vanish at fixed endpoints of the time interval $[0, T]$.
From the above variations and Hamilton's principle, the Euler--Poincar\'e equations with the advected parameter in $V^*$ and the additional dynamics in $Q$ can be derived as follows:
\begin{equation}\label{eq: conEPwithAD}
    \begin{aligned}
        \frac{d}{dt} \frac{\partial \ell}{\partial \xi}
            & = \ad^*_\xi \frac{\partial \ell}{\partial \xi}
            + \frac{\partial \ell}{\partial a} \diamond a
            + \frac{\partial \ell}{\partial n} \diamond n
            + \frac{\partial \ell}{\partial \nu} \diamond \nu, \\
        \frac{d}{dt} \frac{\partial \ell}{\partial \nu}
            & = \frac{\partial \ell}{\partial n}
            + \mathcal{L}^\ast_{\xi} \frac{\partial \ell}{\partial \nu},
    \end{aligned}
\end{equation}
and the curve $a(t) = g(t)^{-1}a_0$ is again the unique solution of 
\begin{equation}\label{eq:conta}
    \begin{aligned}
 \dfrac{d a}{dt} = - \Liexi a, \ a(0)=a_0.
    \end{aligned}
\end{equation}
The complementary relation of $\xi,n,\nu$ reads
\begin{equation}
    \begin{aligned}
\dfrac{d n}{d t} = \nu - \xi n.
    \end{aligned}
\end{equation}
Here, the operator $\mathcal{L}^\ast$ 
is defined by the pairing 
\begin{equation}\label{eq: dual_Lie_op}
    \left\langle \dfrac{\partial \ell}{\partial \nu}, \Liexi w \right\rangle :=
    \left\langle \Liedualxi \dfrac{\partial \ell}{\partial \nu}, w \right\rangle, \quad \text{for } \xi\in\g, w\in T_nQ,
\end{equation}
and representation of the diamond operator in $\frac{\partial \ell}{\partial n} \diamond n
            + \frac{\partial \ell}{\partial \nu} \diamond \nu$ 
            follows from  that used in \cite{HOLM2023133847}, which is (local) decomposition of the operator   $\diamond: T^*TQ\times TQ \to \gad$ (corresponding to the Lie group $G$ acting on the manifold $Q$) defined as
            \begin{equation}\label{eq: concotdia}
   \left\langle \left(\frac{\partial \ell}{\partial n},\frac{\partial \ell}{\partial \nu}\right), \eta (n, \nu)\right\rangle_{T^*_{(n,\nu)}TQ\times T_{(n,\nu)}TQ} := -\left\langle \left(\frac{\partial \ell}{\partial n},\frac{\partial \ell}{\partial \nu}\right)\diamond (n,\nu), \eta\right\rangle_{\gad \times \g},\quad \text{for } \eta \in \g.
            \end{equation}

\subsection{The continuous Kelvin--Noether theorem}
The Kelvin--Noether theorem combines ideas from Kelvin's circulation theorem and Noether's theorem. It relates the conservation of circulation in a fluid to the underlying symmetries of the fluid's dynamics, particularly those arising from its Lagrangian formulation and the associated Lie-group symmetries \cite{gawlik2011geometric,HOLM19981}.

The Noether's theorem states that every finite-dimensional variational Lie group symmetry gives rise to a conservation law
of the Euler--Lagrange equations (e.g., \cite{noether1918,olver1993}).
Discrete and semi-discrete extensions of this theorem have also been proposed (e.g., \cite{marsden2001discrete, peng2017symmetries,peng2022dd}).
In particular, for Euler--Poincar\'e equations, this extends the Kelvin's circulation theorem for continuum mechanics and is called the Kelvin--Noether theorem (e.g., \cite{cotter2013noether,gawlik2011geometric,HOLM19981}). 
Next we introduce the Kelvin--Noether theorem for Euler--Poincar\'e equations with advected parameters and additional dynamics.
Note that in the current study,  the Lie group $G$ is assumed to be finite-dimensional, and hence $\g^{**} \cong \g$.
\begin{thm}\label{thm: conKNthm}
    Let $G$ be a group acting on a manifold $\mathcal{C}$ and a vector space $V$ both from the left.  
    Consider an equivalent map $\mathcal{K}: \mathcal{C} \times V^* \to \g^{**} \cong \g$, that is,
    \begin{equation}
        \mathcal{K}(gc, ga) = \Ad_{g} \mathcal{K}(c, a) 
    \end{equation}
    for any $g \in G$, $c \in \mathcal{C}$ and $a \in V^*$.
    The adjoint action $\Ad: G \times \g \to \g$ is defined by the tangent map of the left and right multiplications $L_g$ and $R_g$ (for $g\in G$) at the identity $e$ as
    \begin{equation}\label{eq: ad_action}
        \Ad_g := T_e \left(L_g \circ R_{g\inverse} \right).
    \end{equation}
    Furthermore, consider a path $( \xi(t), n(t), \nu(t), a(t))\in \g\times TQ \times V^*$ satisfying the  Euler--Poincar\'e equations \eqref{eq: conEPwithAD} and the advected parameter  dynamics \eqref{eq:conta}. 
    For $c(t) = g(t)^{-1} c_0$ with $c_0 \in \mathcal{C}$ and $a(t) = g(t)^{-1} a_0$, define the (left-left) Kelvin--Noether quantity $\mathcal{I}: \mathcal{C} \times \g \times TQ \times V^* \to \bbbr$ by
    \begin{equation}\label{eq: KN_quantity}
        \mathcal{I}(c, \xi, n, \nu, a) = \left\langle \mathcal{K}(c, a), \frac{\partial \ell}{\partial \xi} - \frac{\partial \ell}{\partial \nu} \diamond n \right\rangle,
    \end{equation}
    which satisfies that
    \begin{equation}
        \frac{d}{dt} \mathcal{I}(t) = \left\langle \mathcal{K}(c(t), a(t)), \frac{\partial \ell}{\partial a} \diamond a \right\rangle.
    \end{equation}
\end{thm}
\begin{proof}
    The proof is similar to that in \cite{HOLM19981, stoica}, with an extension to account for the additional dynamics. 
    Recall that $g(t)$, $\xi(t)= g\inverse(t)\dot{g}(t)$, $n(t)=g\inverse(t)q(t)$ and $\nu(t)=g\inverse(t)\dot{q}(t)$ are time-dependent.  
    For any $\mu \in \g$,  we have
    \begin{equation}
        \begin{aligned}
            \left\langle \frac{d}{dt} \Ad^*_{g\inverse} \left( \frac{\partial \ell}{\partial \nu} \diamond n \right), \mu\right\rangle 
            & = \frac{d}{dt} \left\langle - \frac{\partial \ell}{\partial \nu}, \left( \Ad_{g\inverse} \mu \right)n\right\rangle \\
            & = \left\langle - \frac{d}{dt} \frac{\partial \ell}{\partial \nu}, \left( \Ad_{g\inverse} \mu \right)n\right\rangle + \left\langle - \frac{\partial \ell}{\partial \nu}, \frac{d}{dt} \left(\Ad_{g\inverse} \mu \right) n \right\rangle \\
            & = \left\langle \Ad^*_{g\inverse}\left( \frac{d}{dt} \frac{\partial \ell}{\partial \nu} \diamond n \right), \mu \right\rangle + \left\langle - \frac{\partial \ell}{\partial \nu}, - \Liexi \left( \Ad_{g\inverse} \mu \right) n +  \Liederi_{\Ad_{g\inverse} \mu} \nu \right\rangle \\
            & = \left\langle \Ad^*_{g\inverse}\left( \frac{d}{dt} \frac{\partial \ell}{\partial \nu} \diamond n \right), \mu \right\rangle + \left\langle \Ad^*_{g\inverse} \left( -\Liedualxi  \frac{\partial \ell}{\partial \nu} \diamond n + \frac{\partial \ell}{\partial \nu} \diamond \nu\right), \mu \right\rangle \\
            & = \left\langle \Ad^*_{g\inverse}\left( \frac{d}{dt} \frac{\partial \ell}{\partial \nu} \diamond n - \Liedualxi  \frac{\partial \ell}{\partial \nu} \diamond n + \frac{\partial \ell}{\partial \nu} \diamond \nu \right), \mu \right\rangle. \\
        \end{aligned}
    \end{equation}
    Taking the derivative of the Kelvin--Noether quantity \eqref{eq: KN_quantity} and taking the Euler--Poincar\'e equations \eqref{eq: conEPwithAD} into consideration, we obtain
    \begin{equation}
        \begin{aligned}
            \frac{d}{dt} \mathcal{I}(t) 
            & = \frac{d}{dt} \left\langle \mathcal{K}(c_0, a_0), \Ad^*_{g\inverse} \left( \frac{\partial \ell}{\partial \xi} - \frac{\partial \ell}{\partial \nu} \diamond n \right) \right\rangle \\
            & = \left\langle \mathcal{K}(c(t), a(t)), \left( \frac{d}{dt} \frac{\partial \ell}{\partial \xi} - \ad^*_{\xi} \frac{\partial \ell}{\partial \xi} - \frac{d}{dt} \frac{\partial \ell}{\partial \nu} \diamond n + \Liedualxi  \frac{\partial \ell}{\partial \nu} \diamond n - \frac{\partial \ell}{\partial \nu} \diamond \nu \right) \right\rangle \\
            & = \left\langle \mathcal{K}(c, a), \frac{\partial \ell}{\partial a} \diamond a \right\rangle,
        \end{aligned}
    \end{equation}
    which completes the proof.
    Note that we have used the formula for differentiating the coadjoint action \cite{stoica,marsden2013introduction}, namely,
    \begin{equation}
        \begin{aligned}
            \frac{d}{dt} \left\langle \Ad^*_{g\inverse} \frac{\partial \ell}{\partial \xi}, \mu\right\rangle 
            & = \frac{d}{dt} \left\langle \frac{\partial \ell}{\partial \xi}, \Ad_{g\inverse} \mu\right\rangle \\
            & = \left\langle \frac{d}{dt} \frac{\partial \ell}{\partial \xi}, \Ad_{g\inverse} \mu\right\rangle + \left\langle \frac{\partial \ell}{\partial \xi},- \ad_\xi \left( \Ad_{g\inverse} \mu \right)\right\rangle \\
            & = \left\langle  \Ad^*_{g\inverse} \left( \frac{d}{dt} \frac{\partial \ell}{\partial \xi} - \ad^*_{\xi} \frac{\partial \ell}{\partial \xi} \right), \mu\right\rangle .
        \end{aligned}
    \end{equation}
\end{proof}

\section{The discrete Euler--Poincar\'e equations with advected parameters and additional dynamics}
\label{sec: DisEPreduction}
In this section, we show the discrete Euler--Poincar\'e equations with advected parameters and additional dynamics.
Before showing its derivation by the discrete variational principle, let us recall the group difference map and some of its properties,  following \cite{bou2009hamilton,bou2007hamilton,gawlik2011geometric}. 
\begin{definition}(Group difference map.)
    A local diffeomorphism $\tau: \g \to G$ that maps a neighborhood $\mathcal{N}$ of $0 \in \g$ to a neighborhood of identity $e \in G$, such that $\tau(0) = e$ and $\tau^{-1}(\xi) = \tau(-\xi)$ for all $\xi \in \mathcal{N}$, is called a group difference map.
\end{definition}

\begin{definition}(Right-trivialized tangent.)
    The right-trivialized tangent map $d\tau: \g \times \g \to \g$ of a group difference map $\tau$ is defined by
    \begin{equation}
        D \tau (\xi) \cdot \chi = R_{\tau(\xi)} d\tau_{\xi}(\chi).
    \end{equation}
\end{definition}
 This was  defined in  \cite{iserles2000lie} (Definition 2.19 therein), which  describes the differential of $\tau$ as a combination of a map $d\tau$ and the right transport of the tangent vector from $e$ to $\tau(\xi)$. 
    Moreover, the right-trivialized tangent satisfies the following relation
    \begin{equation}
        d\tau_{\xi}(\chi) = \Ad_{\tau(\xi)} d\tau_{-\xi}(\chi).
    \end{equation}

\begin{definition}(Inverse right-trivialized tangent.)
    The inverse right-trivialized tangent $d\tau^{-1}: \g \times \g \to \g$ of a group difference map $\tau$ is defined by 
    \begin{equation}
        D\tau^{-1}\left( \tau(\xi)\right) \cdot \chi = d\tau^{-1}_{\xi} \left(R_{\tau(-\xi)} \chi\right).
    \end{equation}
\end{definition}
This definition describes that the differential of $\tau^{-1}$ can be decomposed into the right transport of the tangent vector from $\tau(\xi)$ to $e$ and the map $d\tau^{-1}_{\xi}$. 
    The inverse right-trivialized tangent satisfies
    \begin{equation}\label{eq: dtau_Ad}
        d\tau^{-1}_{\xi} (\chi) = d\tau^{-1}_{-\xi} (\text{Ad}_{\tau(-\xi)} \chi).
    \end{equation}
For details of the above properties, see \cite{bou2007hamilton, iserles2000lie}. 

Now we consider reduced discrete Lagrangians
and derive the discrete Euler--Poincar\'e equations with  advected parameters and  addition dynamics,
which extends the results of \cite{gawlik2011geometric}.
Let $\{t_k\}_{k=0}^N $ be the discretized time for $t\in[0,T]$ and $h:= t_{k+1} - t_k=T/N$ be a fixed time step.
Consider discrete series $\{g_k = g\left( t_k \right)\}$ on the Lie group $G$ and $\{q_k = q\left( t_k \right)\}$ on the additional configuration manifold $Q$.
 (Left) $G$-invariance of the discrete Lagrangian $L^d: \left( G\times G \right) \times \left( Q\times Q \right)  \times V^* \rightarrow \mathbb{R}$, namely
 \begin{equation}
      L^d(\tilde{g}g_k,\tilde{g}g_{k+1},\tilde{g}q_k,\tilde{g}q_{k+1},\tilde{g}a_0) = L^d(g_k,g_{k+1},q_k,q_{k+1},a_0) \text{ for any } \tilde{g}\in G, 
 \end{equation}
 defines
the (left) reduced discrete Lagrangian $\ell^d: \g \times  (Q \times Q) \times V^* \rightarrow \mathbb{R}$ by taking $\tilde{g}=g^{-1}_k$,
\begin{equation}\label{eq: dis_lag}
    \begin{aligned}
        \ell^d \left(\xi_k, n_k, s_k,  a_k \right) = L^d \left(g_k^{-1} g_k, \xi_k, n_k, s_k,  a_k  \right),
    \end{aligned}
\end{equation}
denoted by  $\ell^d_k$,
with
\begin{equation} \label{eq: dis_nk}
    n_k := g_k^{-1} q_{k},\quad 
 s_k := g_k\inverse q_{k+1}, \quad  
    \xi_k := \frac{1}{h}\tau^{-1}(g_k^{-1} g_{k+1}), 
\end{equation}
and the advected parameter
\begin{equation} \label{eq: dis_advected}
    a_k := g_k^{-1} a_0 \in V^*.
\end{equation}
Note that the $\xi_k$ defined in \eqref{eq: dis_nk}, which can be equivalently written as
    \begin{equation}
 g_{k+1} = g_k \tau(h \xi_k),
    \end{equation}
    is the forward Euler formula on Lie groups for $\xi=g^{-1}\dot{g}\in \g$ introduced in the continuous setting (e.g., \cite{bou2009hamilton,bou2007hamilton}). 


The corresponding discrete functional $\mathcal{S}^d_{a_0}: G^{N+1}\times Q^{N+1} \times Q^{N}  \to \mathbb{R}$ reads
\begin{equation} \label{eq: dis_re_AS}
    \mathcal{S}^d_{a_0} \left(\{g_k\}_{k=0}^N, \{n_k\}_{k=0}^N, \{s_k\}_{k=0}^{N-1} \right) = \sum_{k=0}^{N-1} \ell^d ( \xi_k,n_k, s_k, a_k) h = \sum_{k=0}^{N-1} \ell^d_k h.
\end{equation}
%
The discrete variational principle is expressed as $\delta \mathcal{S}^d_{a_0} = 0$ 
for arbitrary variations $\delta g_k$ and $\delta q_k$, and subject to fixed endpoint conditions $\delta g_0 = \delta g_N =0$ and $\delta q_0 = \delta q_N = 0$.
This leads to the following discrete Euler--Poincar{\'e} equations through discrete variational calculus.
\begin{thm} \label{thm: disEPeqwithAD}
    Let $G$ act from the left on the dual vector space $V^*$ and let the discrete Lagrangian $L^d: (G \times G) \times (Q\times Q) \times  V^* \to  \mathbb{R}$ be left $G$-invariant. 
    Let $\ell^d: \g \times (Q \times Q) \times  V^* \rightarrow \mathbb{R}$ be the reduced Lagrangian defined by \eqref{eq: dis_lag}. 
    Suppose $\{  \xi_k, n_{k}, s_k,  a_k\}_{k=0}^{N-1}$ is a discrete series satisfies $\delta \mathcal{S}^d_{a_0} = 0$  
    under variations $\delta g_k$ and $\delta q_k$ with $\delta g_0 = \delta g_N = 0$ and $\delta q_0 = \delta q_N = 0$. 
    Then we  obtain  the (left-left) discrete Euler--Poincar\'e equations with advected parameters and additional dynamics:
\begin{equation}\label{eq: disEPwithAD}
        \begin{aligned}
        \left( d \tau_{h\xi_k}^{-1} \right)^\ast \frac{\partial \ell^d_k}{\partial \xi_k} &  = \left( d \tau_{-h\xi_{k-1}}^{-1} \right)^\ast \frac{\partial \ell^d_{k-1}}{\partial \xi_{k-1}} 
            + h \dfrac{\partial \ell^d_k}{\partial a_k} \diamond a_k 
            + h \frac{\partial \ell^d_k}{\partial n_k} \diamond n_k 
            + h \dfrac{\partial \ell^d_k}{\partial s_k} \diamond s_k , \\
           \dfrac{\partial \ell^d_k}{\partial n_k} & = -\tau^*\left(h \xi_{k-1}\right)  \dfrac{\partial \ell^d_{k-1}}{\partial s_{k-1}}, 
        \end{aligned}
    \end{equation}
    together with the complementary relation
    \begin{equation}
        n_{k+1} =  \tau\left(- h \xi_k\right) s_k,
    \end{equation}
    and $\{a_k\}$ is given by the series   
    \begin{equation}\label{eq:disca}
        \ a_{k+1} = \tau(-h \xi_k) a_k,\quad a_0=a_0.
    \end{equation}
    Analogous to the continuous case \eqref{eq: concotdia}, the diamond operator $\diamond: T^*Q \times Q \to \gad$  corresponding to the Lie group G acting on the manifold $Q$  in the discrete setting is again defined pointwisely in $Q$ as
    \begin{equation}
        \left\langle \frac{\partial \ell^d_k}{\partial n_k}, \mu n_k \right\rangle _{T_{n_k}^*\! Q \times T_{n_k} \! Q}
        := - \left\langle \frac{\partial \ell^d_k}{\partial n_k} \diamond n_k, \mu \right\rangle _{\gad \times \g}, \ \mathrm{for} \  n_k\in Q,\ \dfrac{\partial \ell^d_k}{\partial n_k}  \in T_{n_k}^*\!Q,\ \mathrm{and}\ \mu \in \g.
    \end{equation}   
    The diamond operator in $\dfrac{\partial \ell^d_k}{\partial s_k} \diamond s_k$ is defined in a similar way but at another point $s_k\in Q$. Furthermore, the action $\tau^*(h\xi_{k-1})$ is defined by 
    \begin{equation}
       \left\langle \frac{\partial \ell_{k-1}^d}{\partial s_{k-1}}, \tau(h\xi_{k-1})w_k \right\rangle_{T^*_{s_{k-1}}Q\times T_{s_{k-1}}Q} := \left\langle \tau^*(h\xi_{k-1}) \frac{\partial \ell_{k-1}^d}{\partial s_{k-1}}, w_k\right\rangle_{T^*_{s_{k-1}}Q\times T_{s_{k-1}}Q}, 
    \end{equation}
   for $w_k\in T_{s_{k-1}}Q$  or $T_{n_k}Q$. 
\end{thm}
\begin{proof}
    First, the variation of $\xi_k$ is obtained by using \eqref{eq: dtau_Ad} as \cite{ bou2007hamilton,gawlik2011geometric}
    \begin{equation}
        \begin{aligned}
            \delta \xi_k 
            = d \tau_{-h \xi_k}^{-1} (\eta_{k+1}) - d \tau_{h \xi_k}^{-1} (\eta_k) .
        \end{aligned}
    \end{equation} 
    Introducing the variations $\eta_k := g_k\inverse \delta g_k$ and  $w_k := g_k\inverse \delta q_k$, the variations of $n_k,\ s_k$ and $a_k$  can be calculated using the relations \eqref{eq: dis_nk} and \eqref{eq: dis_advected} directly, and we obtain 
    \begin{equation}
        \delta n_k 
       = w_k - \eta_k n_k,\quad 
        \delta s_k = \tau\left(h \xi_k\right) w_{k+1} -\eta_ks_k,\quad 
        \delta a_k = - \eta_k a_k.
    \end{equation}
    Using the fixed endpoint condition, the discrete variational principle then gives
    \begin{equation}
        \begin{aligned}
            0 & 
            = \delta \sum_{k=0}^{N-1} \ell^d_k h  = \sum_{k=0}^{N-1} \left\langle h \frac{\partial \ell^d_k}{\partial a_k}, \delta a_k \right\rangle 
            + \left\langle h \frac{\partial \ell^d_k}{\partial \xi_k}, \delta \xi_k \right\rangle 
            + \left\langle h \frac{\partial \ell^d_k}{\partial n_k}, \delta n_k \right\rangle 
            + \left\langle h \frac{\partial \ell^d_k}{\partial s_k}, \delta s_k \right\rangle \\
            & = \sum_{k=0}^{N-1} \left[ \left\langle h \frac{\partial \ell^d_k}{\partial a_k}, -\eta_k a_k \right\rangle 
            + \left\langle \frac{\partial \ell^d_k}{\partial \xi_k},  d \tau_{-h \xi_k}^{-1} (\eta_{k+1})- d \tau_{h \xi_k}^{-1} (\eta_k)  \right\rangle
            + \left\langle h \frac{\partial \ell^d_k}{\partial n_k},   w_k -\eta_k n_k\right\rangle \right. \\
            &\quad\quad\quad +  \left. \left\langle h \dfrac{\partial \ell^d_k}{\partial s_k},  \tau\left(h \xi_k\right) w_{k+1} - \eta_k s_k \right\rangle \right] \\
            & = \sum_{k=1}^{N-1} \left[ \left\langle h \frac{\partial \ell^d_k}{\partial a_k} \diamond a_k, \eta_k \right\rangle 
            + \left\langle - \left( d \tau_{h\xi_k}^{-1} \right)^\ast \frac{\partial \ell^d_k}{\partial \xi_k} + \left( d \tau_{-h\xi_{k-1}}^{-1} \right)^\ast \frac{\partial \ell^d_{k-1}}{\partial \xi_{k-1}}, \eta_{k} \right\rangle \right.\\
                    & \quad\quad
                    + \left\langle h \frac{\partial \ell^d_k}{\partial n_k}, w_k \right\rangle   
                    + \left\langle h \frac{\partial \ell^d_k}{\partial n_k} \diamond n_k, \eta_k \right\rangle 
                    + \left.\left\langle \tau^* \left(h \xi_{k-1}\right) h \dfrac{\partial \ell^d_{k-1}}{\partial s_{k-1}}, w_{k} \right\rangle    
                    +  \left\langle h \dfrac{\partial \ell^d_k}{\partial s_k} \diamond s_k, \eta_k \right\rangle \right], \\
                & = \sum_{k=1}^{N-1} \left[ \left\langle h \frac{\partial \ell^d_k}{\partial a_k} \diamond a_k - \left( d \tau_{h\xi_k}^{-1} \right)^\ast \frac{\partial \ell^d_k}{\partial \xi_k} 
                    + \left( d \tau_{-h\xi_{k-1}}^{-1} \right)^\ast \frac{\partial \ell^d_{k-1}}{\partial \xi_{k-1}} 
                    + h \frac{\partial \ell^d_k}{\partial n_k} \diamond n_k + h \dfrac{\partial \ell^d_k}{\partial s_k} \diamond s_k,
                    \eta_k \right\rangle \right. \\
        &\quad\quad \quad  + \left. \left\langle h \frac{\partial \ell^d_k}{\partial n_k} + \tau^*\left(h \xi_{k-1}\right) h \dfrac{\partial \ell^d_{k-1}}{\partial s_{k-1}}, w_{k} \right\rangle \right],
        \end{aligned}
    \end{equation}
    which yields the discrete Euler--Poincar\'e equations \eqref{eq: disEPwithAD}. 
\end{proof}

{\bf The discrete Kelvin--Noether theorem.}
The following theorem describes a discrete analogue of the continuous Kelvin--Noether Theorem \ref{thm: conKNthm}. 
Again, we assume that the Lie group $G$ is finite-dimensional, and hence $\g^{**} \cong \mathfrak{g}$.

\begin{thm}\label{thm: discreteKelvinNoether}
    Let $G$ be a Lie group that acts on a manifold $\mathcal{C}$  from the left and 
    suppose $\mathcal{K}: \mathcal{C} \times V^* \to \g^{**} \cong \g$ is an equivalent map, namely
    \begin{equation}
        \mathcal{K}(gc, ga) = \Ad_{g} \mathcal{K}(c, a)
    \end{equation}
    for any $g \in G$, $c \in \mathcal{C}$ and $a \in V^*$. 
    Suppose the discrete series $\{g_k, \xi_k, n_k, s_k, a_k\}_{k=0}^{N-1}$ satisfies the (left-left) discrete Euler--Poincar\'e equations \eqref{eq: disEPwithAD} and the discrete advected parameter dynamics \eqref{eq:disca}. 
    Fix $c_0 \in \mathcal{C}$ and define $c_k = g_k^{-1} c_0$. 
     We then define the discrete (left-left)  Kelvin--Noether quantity $I: \mathcal{C} \times \g \times (Q \times Q) \times V^* \to \bbbr$ by
    \begin{equation}
        \mathcal{I}(c_k, \xi_k, n_k, s_k, a_k) = \left\langle \mathcal{K}(c_k, a_k), (d \tau_{h \xi_k}^{-1})^* \frac{\partial \ell^d_k}{\partial \xi_k} - h\frac{\partial \ell^d_k}{\partial s_k} \diamond   s_k\right\rangle,
    \end{equation}
    where $\tau: \g \to G$  is a given group difference map.
    As a consequence, the discrete  Kelvin--Noether  quantity $\mathcal{I}_k := \mathcal{I}(c_k, \xi_k, n_k, s_k, a_k)$ satisfies
    \begin{equation}
        \mathcal{I}_k - \mathcal{I}_{k-1} = \left\langle \mathcal{K}(c_k, a_k), h \frac{\partial \ell^d_k}{\partial a_k} \diamond a_k\right\rangle.
    \end{equation}
\end{thm}
\begin{proof}
    By equivalence of the map $\mathcal{K}$, we have 
    \begin{equation*}
        \begin{aligned}
            \mathcal{I}_k  - & \mathcal{I}_{k-1}  =  \left\langle \mathcal{K}(c_k, a_k), (d \tau_{h \xi_k}^{-1})^* \frac{\partial \ell^d_k}{\partial \xi_k} -  h\frac{\partial \ell^d_k}{\partial s_k} \diamond  s_k\right\rangle \\
            & \quad \quad 
            - \left\langle \Ad_{g^{-1}_{k-1} g_k} \mathcal{K}(c_{k}, a_{k}), (d \tau_{h \xi_{k-1}}^{-1})^* \frac{\partial \ell^d_{k-1}}{\partial \xi_{k-1}} -  h\frac{\partial \ell^d_{k-1}}{\partial s_{k-1}} \diamond \tau(h \xi_{k-1}) n_{k} \right\rangle \\
            & = \left\langle \mathcal{K}(c_k, a_k), (d \tau_{h \xi_k}^{-1})^* \frac{\partial \ell^d_k}{\partial \xi_k} - \left( d \tau_{-h\xi_{k-1}}^{-1} \right)^\ast \frac{\partial \ell^d_{k-1}}{\partial \xi_{k-1}} \right\rangle \\
            &\quad \quad 
            + \left\langle \mathcal{K}(c_k, a_k), - h\frac{\partial \ell^d_k}{\partial s_k} \diamond  s_k + \tau^*(h\xi_{k-1}) h \frac{\partial \ell^d_{k-1}}{\partial s_{k-1}} \diamond n_{k}  \right\rangle \\
            & = \left\langle \mathcal{K}(c_k, a_k), h \frac{\partial \ell^d_k}{\partial a_k} \diamond a_k + h \frac{\partial \ell^d_k}{\partial n_k} \diamond n_k + h\dfrac{\partial \ell^d_k}{\partial s_k} \diamond s_k\right\rangle 
            + \left\langle \mathcal{K}(c_k, a_k),- h \dfrac{\partial \ell^d_k}{\partial s_k} \diamond s_k - h \frac{\partial \ell^d_k}{\partial n_k} \diamond n_k \right\rangle \\
            & = \left\langle \mathcal{K}(c_k, a_k), h \frac{\partial \ell^d_k}{\partial a_k} \diamond a_k \right\rangle.
        \end{aligned}
    \end{equation*}
    Note that in the second equality, we have utilized the relation \eqref{eq: dtau_Ad}; in the third equality, the discrete Euler--Poincar\'e equations \eqref{eq: disEPwithAD} have been substituted.
\end{proof}

\section{Application to the dynamics of underwater vehicles} \label{sec: Application2UV}

In this section, 
the proposed Euler--Poincar\'e reduction  for systems with advected parameters and additional dynamics is applied to study the behavior of underwater (or submersible) vehicles. 
Underwater vehicle dynamics has been studied from both practical and theoretical perspectives (e.g., \cite{cristi1990, fossen1995,panda2021review}), 
particularly using geometric mechanics (e.g., \cite{leonard1997, amit}). 
We consider underwater vehicles with added mass, also known as hydrodynamic or virtual mass, which affects the inertia matrix of the vehicles. Additionally, we assume that the center of gravity of the underwater vehicle does not coincide with its center of buoyancy.  This is a practical design feature, as underwater vehicles are typically constructed with the center of gravity positioned below the center of buoyancy (i.e., bottom-heavy) to ensure stability \cite{leonard1997}.

In this case, the Lie group $G$ is the 3-dimensional special orthogonal group 
\begin{equation}
    \SO(3) := \left\{R \in M(3,\mathbb{R}) \mid R\trans R = I, \det(R)=1 \right\}
\end{equation} 
and its associated Lie algebra is
\begin{equation}
    \so(3) := \left\{\Omega \in M(3,\mathbb{R}) \mid \Omega = -\Omega\trans \right\},
\end{equation}
where $I$ is the $3\times 3$ identity matrix.
In the dynamics of underwater vehicles, the rotation matrix $R \in \SO(3)$ represents the attitude of the body and also serves as a map from a reference configuration to the body configuration.
The differential of $R\trans R = I$ with respect to $t$ reads 
\begin{equation}
    \left( R\trans \dot{R} \right)\trans + R\trans \dot{R} =0.
\end{equation}
This implies that the body angular velocity matrix  $\Omega = R\trans \dot{R} = R\inverse \dot{R} \in \so(3)$, which is related to the body angular velocity vector $\omega \in \bbbrrr$
through the Lie algebra isomorphism
\begin{equation} \label{eq: skewmap}
\begin{aligned}
    ()\skewmap:  \left( \bbbrrr, \times \right) &\to \left(\so(3),[\cdot,\cdot] \right)\\
     \omega=(\omega^1,\omega^2,\omega^3)^{\trans}  &\mapsto \omega\skewmap =\Omega=\left(
   \begin{array}{ccc}
   0 & -\omega^3 & \omega^2\\
   \omega^3 & 0 & -\omega^1\\
   -\omega^2 & \omega^1 & 0
   \end{array}
   \right).
    \end{aligned}
\end{equation}
In this paper, we follow \cite{gawlik2011geometric} and define the pairing between $\so(3)$ and its dual as $\left\langle \Omega, \Psi \right\rangle := \dfrac12 \tr\left(\Omega\trans \Psi\right)$ such that 
$\left\langle u,v \right\rangle = \left\langle u\skewmap, v\skewmap\right\rangle$ for any $u,v\in\mathbb{R}^3$,
where $\tr$ denotes the trace of matrices.

{\bf{Continuous case.}} In the dynamics of underwater vehicles, the kinetic energy $K$ is expressed as
\begin{equation}
    K = \frac{m}{2} \langle \dot{q}, \dot{q} \rangle + \frac{1}{2} \dot{q}\trans R M_A R\trans \dot{q} + \frac{1}{2} \tr \left( \dot{R} \hat{J} \dot{R}\trans \right),
\end{equation}
where the first and second terms come from linear motion of vehicles and the hydrodynamic added mass, and the third term represents the rotational kinetic energy. Inner product of vectors in Euclidean spaces is denoted by $\langle \cdot, \cdot \rangle$.
The potential energy $U$ corresponding to the gravitational and buoyancy restoring forces is given by
\begin{equation}
    U = - (m\mathrm{g} - \rho |\mathcal{V}|\mathrm{g}) \langle q, e_z \rangle+ \rho |\mathcal{V}|\mathrm{g} \langle R r, e_z \rangle,
\end{equation}
where the term $-m\mathrm{g} \langle q, e_z \rangle$ represents the gravitational potential energy and the remaining term represent the buoyant potential energy.
Here, we denote 
the vehicle's mass by $m \in \mathbb{R}^+$, its volume by $|\mathcal{V}|$, the hydrodynamic added mass matrix by $M_A \in M(3,\mathbb{R})$ assumed to be symmetric for simplicity, and the fluid's density by $\rho$.
The constant  vector $r\in \bbbrrr$ denotes the position of the center of buoyancy as seen from the (center of the gravity of) vehicle, and the gravitational constant is $\mathrm{g}$. 
The vector $e_z$ is the unit vector along the vertical axis of the space-fixed frame. 
Furthermore, the modified inertia tensor $\hat{J} $ is related to the inertia tensor $J$  in the space-fixed frame  via
\begin{equation}
    \hat{J} =\frac{1}{2} \tr(J)I - J.
\end{equation}
Note that while the inertia tensor $J$ is a symmetric, positive-definite matrix, the modified inertia tensor $\hat{J}$ is not necessary positive-definite. 

The total energy is $E=K+U$, while 
the Lagrangian $L: T(G\times Q) \times V^* \to \mathbb{R} $ is defined by 
\begin{equation} \label{eq: con_lag_UV}
    \begin{aligned}
        L(R, \dot{R}, q, \dot{q}, e_z) & =K-U\\
        &=\frac{m}{2} \langle \dot{q}, \dot{q} \rangle + \frac{1}{2} \dot{q}\trans R M_A R\trans \dot{q} + \frac{1}{2} \tr \left( \dot{R} \hat{J} \dot{R}\trans \right) 
        + (m\mathrm{g} - \rho |\mathcal{V}|\mathrm{g}) \langle q, e_z \rangle - \rho |\mathcal{V}|\mathrm{g} \langle R r, e_z \rangle .
    \end{aligned}
\end{equation}
Here, $G=\SO(3)$, $Q=\mathbb{R}^3$ and  $V^*=\mathbb{R}^3$. 
Left-invariance of the Lagrangian can be immediately checked that
\begin{equation}
    L( \tilde{R} R, \tilde{R}  \dot{R}, \tilde{R}  q, \tilde{R}  \dot{q}, \tilde{R}  e_z) = L(  R,  \dot{R}, q, \dot{q}, e_z) \text{ for } \forall \tilde{R} \in \SO(3),
\end{equation}
and the reduced Lagrangian $\ell: \g \times TQ \times  V^* \to \mathbb{R} $ is defined  by taking $\tilde{R} =R^{-1}$ in the action as
\begin{equation}\label{eq: con_red_lag_UV}
    \begin{aligned}
        \ell  (\Omega, n, \nu, a) & := L(R^{-1} R, R^{-1} \dot{R}, R^{-1} q, R^{-1} \dot{q},  R^{-1} e_z) \\
        & ~ = \frac{m}{2} \langle \nu, \nu \rangle + \frac{1}{2} \nu\trans M_A \nu + \frac{1}{2} \tr \left( \Omega \hat{J} \Omega\trans \right) + (m\mathrm{g} - \rho |\mathcal{V}|\mathrm{g}) \langle n, a\rangle  - \rho |\mathcal{V}|\mathrm{g} \langle r, a \rangle, \\
    \end{aligned}
\end{equation}
where $\Omega = R^{-1} \dot{R}, \ n = R^{-1} q,\ \nu = R^{-1} \dot{q}$, and the advected parameter is $a = R^{-1} e_z$. 

Defining $w =  R^{-1} \delta q $ and $\eta = R^{-1} \delta R$, we can derive the variations of $\Omega,\ n,\ \nu, \andrm a$, analogous to \eqref{eq:variadd} as follows:
\begin{equation} \label{eq: vari_EP}
    \begin{aligned}
      \delta \Omega & 
= \dot \eta + \ad_\Omega \eta, \quad
        \delta n  = w - \eta n , \quad
        \delta \nu 
        = \dot w + \mathcal{L}_{\Omega} w - \mathcal{L}_{\eta} \nu, \quad
        \delta a  = 
        - \eta a  =  -\mathcal{L}_{\eta} a.
    \end{aligned}
\end{equation}
Define the reduced functional as
\begin{equation}\label{eq: reduced_functional}
    \mathcal{S} = \int_{0}^{T} \ell  (\Omega, n, \nu,  a) dt,
\end{equation}
and the Euler--Poincar\'e equations \eqref{eq: conEPwithAD} for underwater vehicle dynamics can directly be  calculated, which read
\begin{equation}  \label{eq: continuousEPforUV}
    \begin{aligned}
         J \dot \omega & = \left(J \omega\right) \times \omega  - \rho |\mathcal{V}|\mathrm{g} (r \times a ) + \left(mI + M_A\right)\nu \times  \nu , \\
        \left(mI + M_A\right) \dot \nu & = (m\mathrm{g} - \rho |\mathcal{V}|\mathrm{g})a - \omega \times \left(mI + M_A\right) \nu, \\
    \end{aligned}
\end{equation}
and $a$ satisfies $\dot{a}=-\omega\times a$ with initial values $a(0)=e_z$,
where $\omega\skewmap = \Omega$. The complementary relation between $\omega,n,\nu$ reads $\dot n =\nu -\omega \times n$.
The following lemma is applied here (e.g., \cite{gawlik2011geometric}).

\begin{lemma}
For any $u, v\in \mathbb{R}^3$, 
\begin{itemize}
    \item  $u\skewmap v=u\times v$;
    \item the diamond operator $\diamond: \mathbb{R}^3 \times \mathbb{R}^3 \to \so(3)$ is given by $u\diamond v = (u \times v) \skewmap=[u\skewmap,v\skewmap]$.
\end{itemize}
\end{lemma}

{\bf{Discrete case.}}
We discretize the translational velocity $\dot q_k$ by the forward difference 
\begin{equation}
    \dot q_k = \frac{q_{k+1}-q_k}{h},
\end{equation}
and the relation $\dot{R}= R \Omega$ by the forward Euler formula
\begin{equation} \label{eq: LieEuler}
    R_{k+1} = R_k \tau(h\Omega_k),
\end{equation}
and hence $\Omega_k=\dfrac{1}{h}\tau^{-1}\left(R_k^{-1}R_{k+1}\right)$.
The discrete kinetic energy is given by
\begin{equation}
    \begin{aligned}
    K_k = \frac{m}{2} \left\langle \frac{q_{k+1} - q_k}{h},\frac{ q_{k+1} - q_k}{h} \right\rangle 
    + \frac{1}{2} \left(\frac{q_{k+1} - q_k}{h}\right)\trans R_kM_A R_k^{\trans} \left(\frac{q_{k+1} - q_k}{h}\right) 
    + \frac{1}{2} \tr \left( \Omega_k\hat{J} \Omega_k\trans \right) ,
    \end{aligned}
\end{equation}
while the discrete potential energy $U_k$ is given by 
\begin{equation}
    U_k = - \left(m\mathrm{g} - \rho |\mathcal{V}|\mathrm{g}\right) \langle q_k, e_z \rangle + \rho |\mathcal{V}|\mathrm{g} \langle R_k r, e_z \rangle.
\end{equation}
This gives the discrete total energy $E_k = K_k + U_k$, and the discrete Lagrangian 
\begin{equation}
    L^d: (G\times G) \times (Q\times Q) \times V^* \to \mathbb{R},\quad L^d  \left( R_k, R_{k+1}, q_k, q_{k+1}, e_z \right) = K_k - U_k.
\end{equation}
Invariance of the discrete Lagrangian defines a reduced discrete Lagrangian $\ell^d: \g \times (Q \times Q) \times  V^* \to \mathbb{R}$, denoted by $\ell^d_k=\ell^d \left( \Omega_k,  n_k, s_k, a_k \right)$, as follows
\begin{equation}\label{eq: dis_lag_UV}
    \begin{aligned}
         \ell^d \left( \Omega_k,  n_k, s_k, a_k \right) 
        & = L^d \left( R_k\inverse R_k, R_k\inverse R_{k+1}, R_k\inverse q_k, R_k\inverse q_{k+1}, R_k\inverse e_z \right)\\
        & = \frac{m}{2} \left\langle \frac{s_k - n_k}{h}, \frac{s_k - n_k}{h} \right\rangle 
        + \frac{1}{2} \left(\frac{s_k - n_k}{h}\right)\trans M_A \left(\frac{s_k - n_k}{h}\right)
        + \frac{1}{2} \tr \left( \Omega_k \hat{J} \Omega_k \trans \right) \\
        & \quad 
        + \left(m\mathrm{g} - \rho |\mathcal{V}|\mathrm{g}\right) \left\langle n_k, a_k \right\rangle - \rho |\mathcal{V}|\mathrm{g} \langle r, a_k \rangle,
    \end{aligned}
\end{equation}
where 
\begin{equation}
    n_k =  R_k^{-1} q_k, \ s_k = R_k^{-1}q_{k+1}, \ \Omega_k = \dfrac{1}{h}\tau\inverse\left(R_k^{-1}R_{k+1}\right),\ a_k = R_k^{-1} e_z.
\end{equation}
Substituting the Lagrangian \eqref{eq: dis_lag_UV} to \eqref{eq: disEPwithAD} amounts to  the discrete Euler--Poincar\'e equations for the underwater vehicle dynamics as follows: 
\begin{equation} \label{eq: discreteEPforUV}
    \begin{aligned}
        & \left( d \tau_{h\Omega_k}^{-1} \right)^\ast \Pi_{k} = \left( d \tau_{-h\Omega_{k-1}}^{-1} \right)^\ast \Pi_{k-1}
        - h \rho |\mathcal{V}|\mathrm{g} \left( r \times a_k\right)\skewmap 
        + \frac{1}{h} \left(  \left(mI + M_A\right)\left(s_k - n_{k}\right) \times\left(s_k - n_{k}\right) \right)\skewmap,\\
        & \left(mI + M_A\right)\left(s_k - n_{k}  \right)
        = \tau\left(- h\Omega_{k-1}\right) \left(mI + M_A\right) \left( s_{k-1} - n_{k-1}\right) 
        + h^2 \left(m\mathrm{g} - \rho |\mathcal{V}| \mathrm{g} \right) a_k,
    \end{aligned}
\end{equation}
together with $n_{k+1} =  \tau\left(- h \Omega_k\right) s_k$,
and the advected parameter satisfies $a_{k+1}=R_{k+1}^{-1}R_ka_k=\tau(-h\Omega_k)a_k$,
 where $\Pi_k = \hat{J} \Omega_k + \Omega_k \hat{J} \in \so(3)$.

\begin{remark}
Since the additional dynamics is in the linear space $Q=\mathbb{R}^3$, we may introduce the quantity $\nu_k = \dfrac{s_k - n_k}{h}$, and the equations above can be rewritten as 
\begin{equation} \label{eq: discreteEPforUV_with_nu}
    \begin{aligned}
         \left( d \tau_{h\Omega_k}^{-1} \right)^\ast \Pi_{k} &= \left( d \tau_{-h\Omega_{k-1}}^{-1} \right)^\ast \Pi_{k-1}
        - h \rho |\mathcal{V}|\mathrm{g} \left( r \times a_k\right)\skewmap 
        + h\left( \left(mI + M_A\right)\nu_k \times \nu_k  \right)\skewmap, \\
         \left(mI + M_A\right)\nu_k
        &= \tau\left(- h\Omega_{k-1}\right) \left(mI + M_A\right) \nu_{k-1}
        + h \left(m\mathrm{g} - \rho |\mathcal{V}| \mathrm{g} \right) a_k .
    \end{aligned}
\end{equation}
\end{remark}

{\bf{Discrete Kelvin--Noether quantities.}}
Let $G = \SO(3)$ act on $\mathcal{C} := \so(3)$ through the adjoint representation 
\begin{equation}
  (R_k^{-1},c_0) \mapsto c_k = R_k^{-1} w_0\skewmap R_k \in \mathcal{C}  
\end{equation}  for  $c_0=w_0\skewmap\in \so(3)$ where $w_0 \in \bbbrrr$.
Let $\mathcal{K}: \so(3) \times V^* \to \so(3)$ ($V^*=\mathbb{R}^3$) be the equivariant map $(c_k, a_k) \mapsto c_k$. 
Define
\begin{equation}
\begin{aligned}
     \mathcal{I}(c_k, \xi_k, n_k, s_k, a_k) &=\left\langle  R_k^{-1} w_0\skewmap R_k,\left(d\tau^{-1}_{h \Omega_k}\right)^* \Pi_k  - h\frac{\partial \ell^d_{k}}{\partial s_k} \diamond s_k\right\rangle \\
     &=\left\langle  R_k^{-1} w_0\skewmap R_k,\left(d\tau^{-1}_{h \Omega_k}\right)^* \Pi_k  - \frac{1}{h} \left( (mI+M_A)(s_k-n_k)\times s_k\right)\skewmap\right\rangle  , 
     \end{aligned}
\end{equation}
the discrete Kelvin--Noether Theorem \ref{thm: discreteKelvinNoether} yields
\begin{equation*}\label{eq: exmaple_disKN}
    \begin{aligned}
          \mathcal{I}_k-\mathcal{I}_{k-1} &= \left\langle
          R_k^{-1} w_0\skewmap R_k , h \frac{\partial \ell^d_{k}}{\partial a_k} \diamond a_k
          \right\rangle\\
        &  
        = \left\langle  R_k^{-1} w_0\skewmap R_k, h\left(m\mathrm{g} - \rho |\mathcal{V}|\mathrm{g}\right) \left( n_k\times a_k\right)\skewmap - h\rho |\mathcal{V}|\mathrm{g} \left( r \times a_k\right)\skewmap \right\rangle.
    \end{aligned}
\end{equation*}
Note that $R_k^{-1} w_0^{\skewmap} R_k = \left(R_k^{-1}w_0\right)^{\skewmap}$, which becomes 
$ \left(R_k^{-1}a_0\right)^{\skewmap}=a_k^{\skewmap}$ by choosing $w_0$ as the initial  value $a_0 = e_z$. Consequently, $\mathcal{I}_k-\mathcal{I}_{k-1}=0$, namely,
\begin{equation}
\begin{aligned}
\mathcal{I}_k&=\left\langle  R_k^{-1} e_z\skewmap R_k,\left(d\tau^{-1}_{h \Omega_k}\right)^* \Pi_k  - \frac{1}{h} \left( (mI+M_A)(s_k-n_k)\times s_k\right)\skewmap\right\rangle \\
&=\left\langle   e_z\skewmap, R_k\left(\left(d\tau^{-1}_{h \Omega_k}\right)^* \Pi_k  - \frac{1}{h} \left( (mI+M_A)(s_k-n_k)\times s_k\right)\skewmap \right)R_k^{-1}\right\rangle
\end{aligned}
\end{equation}
is a constant of motion.
Thus, equations \eqref{eq: discreteEPforUV} preserve the $e_z$-component of the vector, which is  the inverse of the following matrix by using the Lie algebra isomorphism \eqref{eq: skewmap}:
\begin{equation}\label{eq:KNq}
    \begin{aligned}
        \pi_k =R_k\left(\left(d\tau^{-1}_{h \Omega_k}\right)^* \Pi_k  - \frac{1}{h} \left( (mI+M_A)(s_k-n_k)\times s_k\right)\skewmap \right)R_k^{-1}.
    \end{aligned}
\end{equation}

Next, let us consider some well-known  group difference maps $\tau$.

{\bf Cayley transform as the group difference map.}
Let the map $\tau$ be the Cayley transform:
\begin{equation}
    \begin{aligned}
        \cay \left(\Omega\right) & = \left( I - \dfrac{\Omega}{2} \right)\inverse \left( I + \dfrac{\Omega}{2} \right), \quad \Omega\in \so(3),
    \end{aligned}
\end{equation}
and the dual of its right-trivialized tangent is given by \cite{gawlik2011geometric}
\begin{equation}
    \begin{aligned}
        \left( d \cay_{h\Omega_k}^{-1} \right)^\ast \Pi_k & = \left(I + \dfrac{h\Omega_k}{2} \right) \Pi_k \left( I - \dfrac{h\Omega_k}{2} \right)= \Pi_{k} + \dfrac{h}{2} \left[\Omega_k, \Pi_{k} \right] - \dfrac{h^2}{4} \Omega_k\Pi_{k}\Omega_k,\quad \Omega_k,\Pi_k\in\so (3) . \\
    \end{aligned}
\end{equation}
Then the discrete Euler--Poincar\'e equations \eqref{eq: discreteEPforUV_with_nu} become
\begin{equation}\label{eq: disEP_cayley_UV_by_nu}
    \begin{aligned}
        & \Pi_k + \frac{h}{2} \left[\Omega_k , \Pi_k \right] -\frac{h^2}{4}  \Omega_{k} \Pi_{k} \Omega_{k}  \\
        & \quad 
        = \Pi_{k-1} - \frac{h}{2} \left[\Omega_{k-1}, \Pi_{k-1} \right] 
        -\frac{h^2}{4}  \Omega_{k-1} \Pi_{k-1} \Omega_{k-1}   
        - h \rho |\mathcal{V}|\mathrm{g} \left( r \times a_k\right)\skewmap 
        + h \left( \left(mI + M_A\right) \nu_k\times \nu_k\right)\skewmap,  \\
        & \left(mI + M_A\right)\nu_k =  \cay\left(- h\Omega_{k-1}\right) \left(mI + M_A\right) \nu_{k-1} + h\left(m\mathrm{g} - \rho |\mathcal{V}| \mathrm{g} \right)  a_k .\\
    \end{aligned}
\end{equation}
To facilitate simulation, we will vectorize the first equation of \eqref{eq: disEP_cayley_UV_by_nu} using the following lemma.

\begin{lemma}
    For any $\omega \in \bbbrrr$, a $3\times 3$ real symmetric matrix $J$ and matrix $\hat{J} = \dfrac{1}{2} \tr \left(J\right) I - J$
    we have the identities 
\begin{equation}\label{eq: relation_omega_J_1}
        \left(J\omega\right)\skewmap = \omega\skewmap \hat{J} + \hat{J} \omega\skewmap ,
    \end{equation}
    \begin{equation}\label{eq: relation_omega_J_2}
        \inpro{J\omega, \omega} = \inpro{\hat{J} \omega\skewmap, \omega\skewmap},
    \end{equation}
    \begin{equation}\label{eq: relation_omega_J_3}
        \left(\omega \times J\omega\right)\skewmap = \left(\omega\skewmap\right)^2 \hat{J} - \hat{J} \left(\omega\skewmap\right)^2,
    \end{equation}
    \begin{equation}\label{eq: relation_omega_J_4}
        \omega\skewmap \left(J\omega\right)\skewmap \omega\skewmap = - \left( \norm{\omega}^2 J \omega + \omega\times \left(\omega \times J \omega \right) \right)\skewmap.
    \end{equation}
\end{lemma}
\begin{proof}
   The identities \eqref{eq: relation_omega_J_1}-\eqref{eq: relation_omega_J_3} can be found in \cite{amit}. Here, we  show proof of the   identity \eqref{eq: relation_omega_J_4}.

Using \eqref{eq: relation_omega_J_1} and \eqref{eq: relation_omega_J_3}, direct calculation gives
\begin{equation}
    \begin{aligned}
        \omega\skewmap \left(J\omega\right)\skewmap \omega\skewmap &=\omega\skewmap \left(  \omega\skewmap \hat{J} + \hat{J} \omega\skewmap \right) \omega\skewmap 
\\
&= \left(\omega\skewmap\right)^2 \hat{J} \omega\skewmap + \omega\skewmap \hat{J} \left(\omega\skewmap\right)^2\\
& = \left(\omega\skewmap\right)^3 \hat{J} + \hat{J} \left(\omega\skewmap\right)^3 - \left[\omega\skewmap, \left(\omega\skewmap\right)^2 \hat{J} - \hat{J} \left(\omega\skewmap\right)^2 \right]\\
    & =  - \norm{\omega}^2 \left( J \omega \right)\skewmap - \left[ \omega\skewmap,  \left(\omega \times J \omega \right)\skewmap  \right] \\
        & = - \left( \norm{\omega}^2 J \omega + \omega\times \left(\omega \times J \omega \right) \right)\skewmap.
    \end{aligned}
\end{equation}
Here, we have used the fact that 
\begin{equation}
    \left(\omega\skewmap\right)^3 = - \norm{\omega}^2 \omega\skewmap.
\end{equation}

\end{proof}

Based on the lemma above, recalling $\Pi_k = \hat{J} \Omega_k + \Omega_k \hat{J}$ and denoting $\Omega_k=\omega_k^{\skewmap}$, the first equation of \eqref{eq: disEP_cayley_UV_by_nu} can be vectorized as follows
\begin{equation}
    \begin{aligned}
        %
        & J \omega_{k} + \dfrac{h}{2} \left(\omega_{k} \times J \omega_{k}\right) + \dfrac{h^2}{4} \left( \norm{\omega_{k}}^2 J \omega_{k} +  \omega_{k} \times \left(\omega_{k} \times J \omega_{k}\right) \right) \\
        & ~~~~~~~ = J \omega_{k-1} - \dfrac{h}{2} \left(\omega_{k-1} \times J \omega_{k-1}\right) + \dfrac{h^2}{4} \left( \norm{\omega_{k-1}}^2 J \omega_{k-1} + \omega_{k-1} \times \left(\omega_{k-1} \times J \omega_{k-1}\right) \right) \\ 
        & \quad \quad \quad \quad 
        - h \rho |\mathcal{V}|\mathrm{g} \left(r \times a_k \right)
        + h \left(mI + M_A\right) \nu_k \times \nu_k , \\
    \end{aligned}
\end{equation}
which is implicit in $\mathbb{R}^3$.
Therefore, we use the Gauss--Newton method for solving  $f: \bbbrrr \to \bbbrrr$,
\begin{equation}
    f(\omega) = J \omega + \dfrac{h}{2} \left(\omega \times J \omega\right) + \dfrac{h^2}{4} \left( \norm{\omega}^2 J \omega+  \omega \times \left(\omega \times J \omega\right) \right).
\end{equation}
Its differential with respect to $\omega$ is
\begin{equation}
    \begin{aligned}
        \frac{\partial}{\partial \omega} f(\omega) = 
        J + \frac{h}{2} \left( \omega \skewmap J - \left(J \omega\right)\skewmap \right)
        + \frac{h^2}{4} \left( 2\omega (J \omega)\trans + \norm{\omega}^2 J - (\omega \times J \omega)\skewmap + \omega \skewmap \left(\omega \skewmap J - \left(J \omega \right)\skewmap\right) \right).
    \end{aligned}
\end{equation}

{\bf Matrix exponential as the group difference map.}
The matrix exponential can also be regarded as the group difference map $\tau$.
On the Lie group  $\SO (3)$, 
the matrix exponential can be written by Rodrigues' formula as the function from $\so(3)$ to $\SO(3)$:
\begin{equation} \label{eq: Rodrigues_formula}
        \exp \left(\omega\skewmap\right) = I + \dfrac{\sin \norm{\omega}}{ \norm{\omega}} \omega\skewmap  + \dfrac{1 - \cos \norm{\omega}}{\norm{\omega}^2} \left(\omega\skewmap \right)^2. \\
\end{equation}
for any $\omega \in \bbbrrr$ and $\omega\skewmap  \in \so(3).$
Its inverse right-trivialized tangent $d\exp\inverse: \so(3) \times \so(3) \to \so(3)$ can be derived as
\begin{equation}
    \begin{aligned}
        d \exp\inverse_{\Omega}\left(\Pi\right) 
        & = \Pi - \frac{1}{2} \left[ \Omega , \Pi \right]
        + \alpha \left[\Omega , \left[\Omega, \Pi \right]\right], \\
    \end{aligned}
\end{equation}
where 
\begin{equation}\label{eq: alpha}
   \alpha = \dfrac{1}{\norm{\omega}^2} \left(1 - \dfrac{\norm{\omega}}{2} \cot \dfrac{\norm{\omega}}{2} \right), \text{ and } \omega\skewmap= \Omega.
\end{equation}
Therefore, we can obtain 
\begin{equation}
    \left( d \exp_{h\Omega_k}^{-1} \right)^\ast \Pi_k = \Pi_k + \frac{h}{2} \left[ \Omega_k , \Pi_k \right] + h^2\alpha_k  \left[\Omega_k , \left[\Omega_k, \Pi_k \right]\right] , \quad \Omega_k,\Pi_k\in\so(3),
\end{equation}
where 
\begin{equation}\label{eq: alpha_k}
    \alpha_k = \dfrac{1}{\norm{\omega_{k}}^2} \left(1 - \frac{\norm{\omega_{k}}}{2} \cot \frac{\norm{\omega_{k}}}{2} \right), \text{ and } \omega_k^{\skewmap}=\Omega_k.
\end{equation}

Consequently, the discrete Euler--Poincar\'e equations \eqref{eq: discreteEPforUV_with_nu} become
\begin{equation}\label{eq: disEP_exp_UV}
    \begin{aligned}
         \Pi_k  + \frac{h}{2} \left[\Omega_k , \Pi_k \right]  &+ h^2\alpha_k \left[\Omega_k , \left[\Omega_k , \Pi_k \right] \right] \\
         &= \Pi_{k-1} - \frac{h}{2} \left[\Omega_{k-1}, \Pi_{k-1} \right] + h^2\alpha_{k-1} \left[\Omega_{k-1}, \left[\Omega_{k-1}, \Pi_{k-1} \right] \right] \\
         & \quad 
         - h \rho |\mathcal{V}|\mathrm{g} \left( r \times a_k\right)\skewmap 
        + h \left( \left(mI + M_A\right)\nu_k \times \nu_k \right)\skewmap, \\
         \left(mI + M_A\right)\nu_k
       & =  \exp\left(- h\Omega_{k-1}\right) \left(mI + M_A\right) \nu_{k-1}
        + h\left(m\mathrm{g} - \rho |\mathcal{V}| \mathrm{g} \right)  a_k .\\
    \end{aligned}
\end{equation}
Similar to the Cayley group difference map case, the first equation of \eqref{eq: disEP_exp_UV} can be vectorized  as 
\begin{equation}
    \begin{aligned}
         J \omega_{k} &+ \dfrac{h}{2} \left(\omega_{k} \times J \omega_{k}\right) + h^2 \alpha_k \left(\omega_{k} \times \left(\omega_{k} \times J \omega_{k}\right)\right) \\
        & \quad \quad  = J \omega_{k-1} - \dfrac{h}{2} \left(\omega_{k-1} \times J \omega_{k-1}\right) + h^2 \alpha_{k-1} \left(\omega_{k-1} \times \left(\omega_{k-1} \times J \omega_{k-1}\right) \right)\\
        & \quad \quad \quad 
        - h \rho |\mathcal{V}|\mathrm{g} \left(r \times a_k\right) 
        + h \left(mI + M_A\right) \nu_k \times \nu_k.
    \end{aligned}
\end{equation}
Again, we use the Gauss--Newton method to implement this implicit scheme by defining a function $f: \bbbrrr \to \bbbrrr$,
\begin{equation}
    f(\omega) = J \omega + \dfrac{h}{2} \left(\omega \times J \omega\right) + h^2 \alpha \left( \omega \times \left(\omega \times J \omega\right)\right),
\end{equation}
where $\alpha$ is given by \eqref{eq: alpha}.
The differential of $f$ with respect to $\omega$ is given by 
\begin{equation}
    \begin{aligned}
        \frac{\partial}{\partial \omega} f(\omega)  = 
        J  & + \frac{h}{2} \left( \omega \skewmap J - \left(J \omega\right)\skewmap \right)\\
        &  
        + h^2 \alpha \left( - (\omega \times J \omega)\skewmap + \omega \skewmap \left(\omega \skewmap J - \left(J \omega \right)\skewmap\right) \right)
        + h^2 \frac{\partial \alpha}{\partial \omega} \left(\omega \times \left(\omega \times J \omega\right)\right)\trans ,
    \end{aligned}
\end{equation}
where
\begin{equation}
    \frac{\partial \alpha}{\partial \omega} = \left( \frac{\norm{\omega}^2}{4} \csc^2 \frac{\norm{\omega}}{2} - 2 + \frac{\norm{\omega}}{2} \cot \frac{\norm{\omega}}{2} \right)\frac{\omega}{\norm{\omega}^4} .
\end{equation}


\section{Numerical simulations} \label{sec: NumericalSimulation}
In this section, 
we perform numerical simulations for underwater vehicle dynamics (see Section \ref{sec: Application2UV}) by using the discrete Euler--Poincar\'e equations. 
In particular, we show the behaviors of the total energy and the discrete Kelvin--Noether quantity. 

The group difference map is chosen as both the Cayley transform and the matrix exponential for the Lie group $\SO(3)$.
The physical quantities used in the simulations are summarized in Table \ref{tab: physical_parameters}, 
where the operator $\diag(x) \in M(n,\mathbb{R})$ returns a diagonal matrix with the given arguments $x \in \mathbb{R}^n$ on its main diagonal components.
Other parameters for numerical simulations are listed in Table \ref{tab: numerical_parameters}.
The initial values of the rotation matrix and the angular velocity are derived from the Euler angles (ZXZ convention), $\left( \psi_0, \theta_0, \phi_0\right)$ and $\left( \dot\psi_0, \dot\theta_0, \dot\phi_0\right)$, respectively.

\begin{table}[htbp]
    \centering
    \caption{ Physical quantities of the underwater vehicle dynamics and their meanings/values (adapted from that used in \cite{amit})}
      \vspace{0.2cm}
    \label{tab: physical_parameters}
    \begin{tabular}{c|l|l}
        \hline
        \textbf{Parameter} & \textbf{Meaning} & \textbf{Value in the simulation} \\ \hline
        $m$                                           & Scalar mass of the vehicle  & 123.8\ kg \\ \hline
        $M_A$                                         & Added mass matrix               & $\diag \left(65, 70,75 \right)$ kg \\ \hline
        $J$                                           & Standard inertia matrix   & $\diag\left(5.46, 5.29, 5.72\right)\ \rm{kg}\cdot\rm{m}^2$\\ \hline
        $\rho \lvert \mathcal{V} \rvert \mathrm{g} $  & Weight of the water displaced by the vehicle & 1215.8 kg\\ \hline
        $\mathrm{g} $                                 & Gravitational constant & 9.81 $\rm{m}/\rm{s}^2$ \\ \hline
        $r$     & Vector from  center of gravity to  center of buoyancy  & $(0,0, -0.007)\trans$ m\\ \hline
    \end{tabular}
\end{table}

\begin{table}[htbp]
    \centering
    \caption{Parameters used in  numerical simulations}
    \vspace{0.2cm}
    \label{tab: numerical_parameters}
    \begin{tabular}{c|l|l}
        \hline
       \textbf{Parameter}
        & \textbf{Meaning} & \textbf{Value in the simulation } \\ \hline
        $h$                                           &  Step size  & 0.01 s\\ \hline
        $q_0$   & Initial value of  position       & $(0, 0, 1)\trans$ m\\ \hline
        $v_0$   & Initial value of  velocity       & $(0.1, 0.1, 0.8)\trans$ $\rm{m}/\rm{s}$ \\ \hline
        $\left( \psi_0,\theta_0,\phi_0 \right)$ & Initial value of  Euler angles (ZXZ convention) & $(2\pi, 0,2\pi)$ rad \\ \hline
        $\left( \dot{\psi}_0, \dot{\theta}_0, \dot{\phi}_0\right)$ & Initial value of Euler-angle rates  & $(10\pi/180,10\pi/180, 10\pi/180)$ rad/s \\ \hline
    \end{tabular}
\end{table}

In the current setting, the underwater vehicle is influenced only by gravity and buoyancy forces and moments, with no other external forces acting on it, and hence the total energy is conserved. 
Figure \ref{fig: relative_ene_error_50_300} shows the relative error of energy over the time span of $[0,500]$ in a semi-logarithmic plot. 
The energy fluctuates within a certain range but tends to increase slightly over time. This is because of the Gauss--Newton method, which brings truncation errors.

The time evolution of the Kelvin--Noether quantity, namely $e_z$-component of the vector associated to $\pi_k$ given by \eqref{eq:KNq}  using the Lie algebra isomorphism \eqref{eq: skewmap}, is shown in Figure \ref{fig: KNquantity_30_500}  as a semi-logarithmic plot. 
Differences can hardly be observed in the two cases. Note that the Cayley transform is a second-order approximation of the matrix exponential \cite{bou2007hamilton}.

Given the current initial data and parameter settings, it is expected that the vehicle will initially rise due to the (positive) initial velocities and then descend under the combined effects of gravitational and buoyancy forces. Figure \ref{fig: Reduced_pos_1000} confirms that the vehicle's trajectories align with these expectations.

\begin{figure}[htbp]
    \centering
    \includegraphics[width=\textwidth]{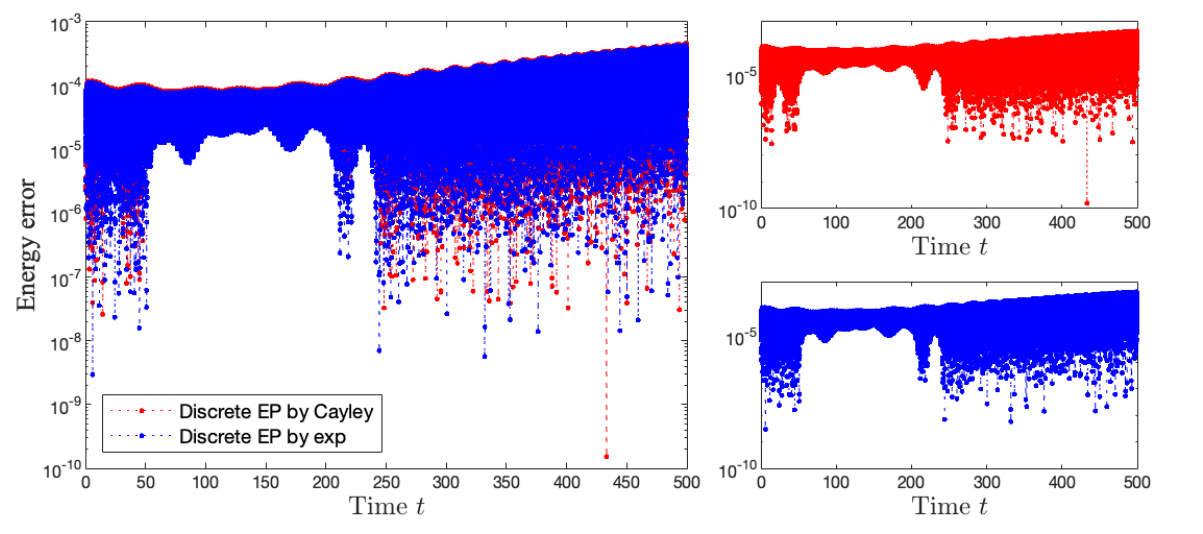}
    \caption{Relative error of the total energy, $\abs{\dfrac{E_k - E_0}{E_0}}$, over time span $[0, 500]$.}
    \label{fig: relative_ene_error_50_300}
\end{figure}

\begin{figure}[htbp]
    \centering
    \includegraphics[width=\textwidth]{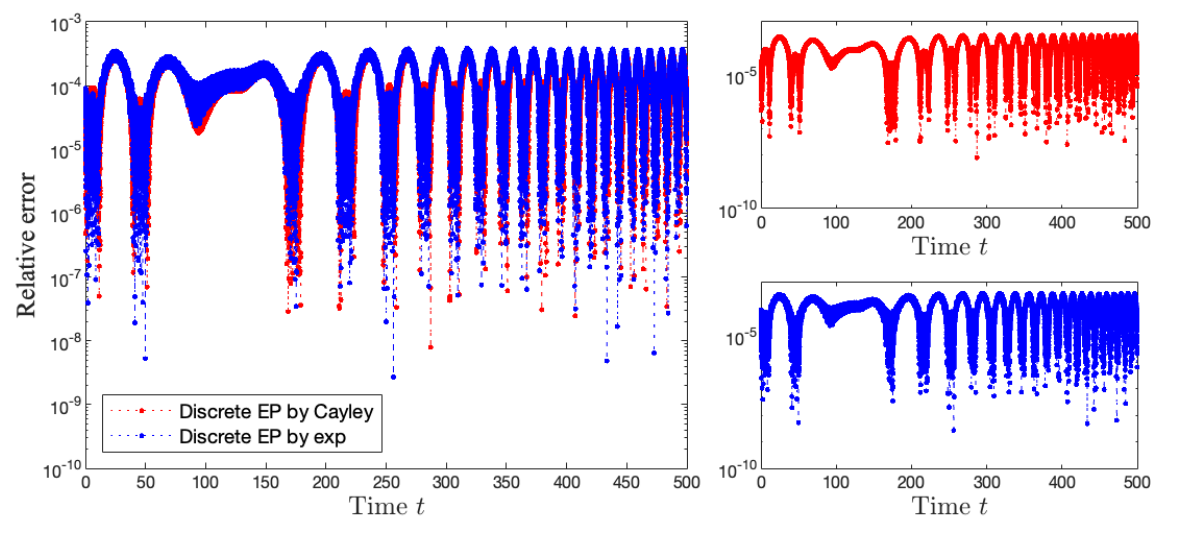}
    \caption{Relative error of the Kelvin--Noether quantity, $\abs{\dfrac{\mathcal{I}_k - \mathcal{I}_0}{\mathcal{I}_0}}$, over time span $[0, 500]$.}
    \label{fig: KNquantity_30_500}
\end{figure}

\begin{figure}[htbp]
    \centering
    \includegraphics[width=\textwidth]{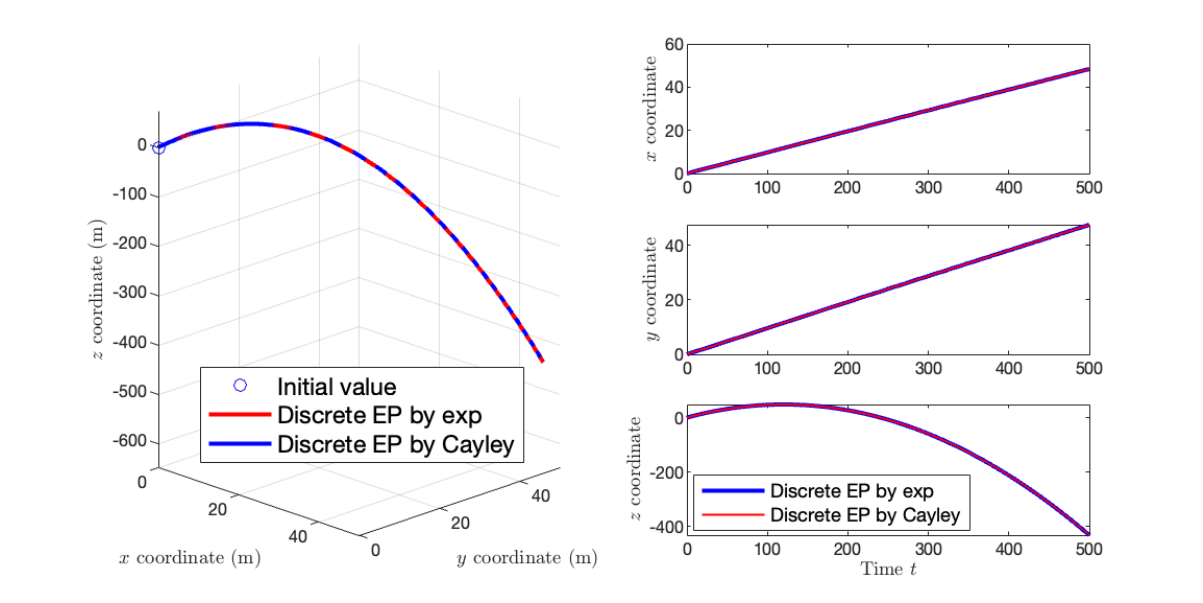}
    \caption{Left: 3-dimensional visualization of the vehicle's trajectory  over time span $[0, 500]$. 
    Right: The trajectory of the vehicle along each axis of the reference frame over time span $[0, 500]$.}
    \label{fig: Reduced_pos_1000}
\end{figure}



\section{Conclusion} \label{sec: Conclusion}
In this paper, we proposed the discrete Euler--Poincar\'e reduction for discrete Lagrangian systems on Lie groups with advected parameters and additional dynamics. The group difference map technique is employed to ensure that the trajectories remain in the configuration space, with the group difference map chosen as either the Cayley transform or the matrix exponential.
Furthermore, we extended the Kelvin--Noether theorem to these systems in both the continuous and the discrete settings.
These methods were applied to study the dynamics of underwater vehicles. We derived both the continuous and discrete Euler--Poincar\'e equations and performed numerical simulations based on the latter. The behavior of the conserved total energy and Kelvin--Noether quantity effectively illustrated the structure-preserving property of the scheme.

In future research, there remains significant scope for further exploration from both theoretical and applied perspectives.
From a theoretical standpoint, it is important to extend the discrete Euler--Poincar\'e reduction for systems on finite-dimensional Lie groups with advected parameters and additional dynamics, as well as the corresponding Kelvin--Noether theorem,  to infinite-dimensional Lie groups. This extension could lead to a more general framework that better captures the complexities of systems with infinite degrees of freedom, such as those arising in fluid dynamics or large-scale mechanical systems. 

Furthermore, as illustrated in simulations of the underwater vehicle dynamics, the fluctuation of total energy increases due to the use of the Gauss--Newton method for the implicit scheme. This issue can be improved by using the moving frame method, which preserves total energy by maintaining the time-translational symmetry of the underlying variational problems, even in variational integration (e.g., \cite{fels1999moving,Mansfield_2010, mv_linyu_2019}). It is important to note that in these numerical schemes, the time step often varies.


From an applied standpoint, fluid velocity and external forces in underwater vehicle dynamics should  be further taken into account for real-world applications, such as control and path planning.
The dynamics of underwater vehicles with fluid velocity has been modeled using the Lagrangian formulation (e.g., \cite{ardakani2019variational, fiori2024}). However, the Euler--Poincaré reduction with additional dynamics has yet to be fully developed for this context. It is expected that fluid velocity can be included in the Lagrangian as an additional advected parameter.
If a second advected parameter is introduced, the Kelvin--Noether theorem will be influenced by both advected parameters, and the associated Kelvin--Noether quantities may not be conserved in certain cases. It would also be interesting to generalize geometric optimal control strategies developed in, for example, \cite{hussein2006discrete,5872067,lee2008optimal}, to invariant systems on Lie groups with advected parameters and additional dynamics using the Euler--Poincar\'e reduction. This would facilitate the extension of optimal control methods to underwater vehicle dynamics and other similar mechanical systems.

\subsection*{Acknowledgments}
The authors are grateful to Darryl Holm, Ruiao Hu and Hiroaki Yoshimura for their inspiring discussions.
YO is partially supported by JST SPRING (JPMJSP2123) and the Keio University Doctorate Student Grant-in-Aid Program from Ushioda Memorial Fund. 
LP is partially supported by JSPS KAKENHI (24K06852), JST CREST (JPMJCR1914, JPMJCR24Q5), and Keio University (Academic Development Fund, Fukuzawa Fund). 

\bibliography{citations} 
\bibliographystyle{abbrv}

\end{document}